\documentclass[12pt]{article}
\usepackage[margin=1in]{geometry}
\usepackage{amsmath}
\usepackage{amssymb}
\usepackage{amsthm}
\usepackage{graphicx}
\usepackage{subfigure}
\usepackage{cite}
\usepackage{url}

\def\hi{({\sf hi})}
\def\lo{({\sf lo})}

\newtheorem{proposition}{Proposition}[section]
\newtheorem{theorem}[proposition]{Theorem}

\title{Modeling and Control of Rare Segments in BitTorrent with Epidemic Dynamics\thanks{A shorter version of this paper that did not include the $N$-segment lumped model was presented in May 2011 at IEEE ICC, Kyoto.}}

\author{
C. Griffin\footnote{C. Griffin is the with Applied Research Laboratory, Penn State University, University Park, PA 16802, E-mail: \texttt{griffinch@ieee.org}}~
G. Kesidis\footnote{G. Kesidis is the with Depts. of Electrical Engineering and Computer Science and Engineering, Penn State University, University Park, PA 16802, E-mail: \texttt{gik2@psu.edu}}~
P. Antoniadis\footnote{P. Antoniadis is with the Communications Systems Group, ETH Zurich, E-mail: \texttt{antoniadis@tik.ee.ethz.ch}} ~and~
S. Fdida\footnote{S. Fdida are with the Computer Science Dept., Univ. Pierre \& Marie Curie (LIP6), \texttt{serge.fdida@lip6.fr}}
}

\begin{document}
\maketitle

\begin{abstract}
Despite its existing incentives for leecher cooperation, BitTorrent file sharing fundamentally relies on the presence of seeder peers. Seeder peers essentially operate outside the BitTorrent incentives, with two caveats: slow downlinks lead to increased numbers of ``temporary" seeders (who left their console, but will terminate their seeder role when they return), and the copyright liability boon that file segmentation offers for permanent seeders. Using a simple epidemic model for a two-segment BitTorrent swarm, we focus on the BitTorrent rule to disseminate the (locally) rarest segments first. With our model, we show that the rarest-segment first rule minimizes transition time to seeder (complete file acquisition) and equalizes the segment populations in steady-state. We discuss how alternative dissemination rules may {\em beneficially increase} file acquisition times causing  leechers to remain in the system longer (particularly as temporary seeders).  The result is that leechers are further enticed to cooperate. This eliminates the threat of extinction of rare segments which is prevented by the needed presence of permanent seeders. Our model allows us to study the corresponding trade-offs between performance improvement, load on permanent seeders, and content availability, which we leave for future work.
Finally, interpreting the two-segment model as one involving a rare segment and a ``lumped" segment representing the rest, we study a model that jointly considers control of rare segments and different uplinks causing ``choking," 
where high-uplink peers will not engage in certain transactions with 
low-uplink peers.
\end{abstract}


\section{Introduction}

There are several different  incentives in the BitTorrent protocol: the segmentation of the data object (file) into pieces\footnote{Alternatively called chunks, segments or blocks in the BitTorrent literature and herein.} to promote swapping of pieces among peers in a swarm, the dissemination strategy of the file pieces (rarest-first),  the uplink reciprocity (choking) strategy when swapping pieces, and the optimistic unchoking strategy. The configuration of these rules can significantly affect performance  under different scenarios and assumptions (e.g., the size of a swarm and of individual  neighborhoods of interacting peers \cite{Legout06}, the amount  of asymmetry between, and distribution of, uplink capacities, etc.).

There is a significant literature  on modeling  the properties of 
the existing BitTorrent algorithm,
e.g.,  \cite{BitTorrent,Cohen:2003,Qiu:2004,Antoniadis:2004,Chow09},
some with an aim to improve its performance. 
Our model is different from those explored in
\cite{Mass-Voj-05,Yang:2004,Qiu:2004} for BitTorrent,
and we compute different quantities of interest. 
In \cite{bittor-imc05}, the authors propose a ``fluid" model
of a single torrent/swarm (as we do in the following) and fit it
to (transient) data drawn from aggregate swarms.
In \cite{Norros09}, they consider a similar model with
normalized terms the effect of which is a nonlinear time-dilation
of the transient dynamics.
The connection to branching
process models \cite{Yang:2004,Ge:2004} is simply
that ours only tracks the number of active peers who possess or
demand the file under consideration, i.e., a single swarm.
In \cite{bitmax} a strategy called BitMax is proposed that can fully use upload capacities of ``resourceful" peers and thus improve performance without the reciprocity strategy implemented today in BitTorrent. 
So, there is a clear tension between maximizing global performance and fairness in BitTorrent~\cite{tradeoffs}.
This means that if certain peers are required to share more resources than they will need to consume themselves, they might choose not to join the system or try to prematurely defect. Studies of incentives primarily focus on reciprocity mechanisms in terms of upload bandwidth, see e.g., \cite{ma:2004,Ma:2004:1} with the objective of fairness. 

Although this is a theoretically interesting question, in practice there are many users that are typically understood  as not behaving rationally on BitTorrent, including the significant number of seeders both for popular and unpopular content~\cite{Bieber06}. The presence of ``permanent" seeders is enabled by  the typically flat-rate pricing (without quotas) for residential Internet access \cite{CISS10}, and the file segmentation itself provides some limitation of liability for illegal dissemination of copyrighted material\footnote{Swarm discovery via third parties, e.g.,  search and downloading torrent files from certain web sites, offers additional limitation of liability for permanent seeders.}. Also, segment extinction is precluded  by the presence of permanent seeders.

It is also well understood that peers spending additional time on-line will improve overall content availability, since while participating in a swarm downloading content, peers do disseminate file pieces belonging to other swarms up to and including the point at which they acquire the entire file and become seeders.

With the presence of \textit{permanent} seeders,  the extinction of rare pieces is not a threat. The presence of temporary seeders is desirably increased by extended download times, as the leecher peers may leave their console while waiting and become seeders while they are absent \cite{Bieber06}.

A delaying strategy may be also implemented by seeders to limit their upload capacity in an ad hoc fashion; a simulation \cite{adar} studied a  seeder strategy to reduce its upload throughput particularly for non-popular items. A general public good model was proposed in \cite{sigecom} focusing on content availability in which the main contribution  of peers is their time on-line instead of upload capacity. 

A main objective of this paper is to study strategies of file-segment dissemination, based on segment rarity, to explore the trade-off between improving download performance and enticing cooperative activity through longer downloading times for leecher peers. That is, deviations from rarest first segment distribution will have the \textit{beneficial} effect of additionally delaying the leechers, all or just some of them, with segment extinction precluded by the presence of permanent seeders. As such, we are interested in the transient behavior of a given swarm, rather than a generic transaction process among a fixed group of peers. The model we use will reflect this emphasis. 

A shorter version of this paper was presented in May 2011 at IEEE ICC, Kyoto. In that paper, we studied only a two-segment model. In this paper, we extend the results of the original paper by considering an $N$ segment model with an intentionally rare final segment as well as two peer classes -- one with a high throughput rate and the other with a low throughput rate. We derive the differential equation model for this case and show how this variation affects sojourn time from leecher to seeder in both classes. 


The remainder of this paper is organized as follows. In Section \ref{sec:EpidemicModel} and \ref{sec:SelectionControl}, we describe a two-piece deterministic epidemic model of a swarm, similar to one previously studied in \cite{NETCOOP08} as a special case of a deterministic limit of a stochastic sequential transactional model; but here we consider a control parameter governing which piece is disseminated given that there is choice. In Section \ref{sec:opt-control-sec}, we argue that the  ``bang-bang" globally rarest first is optimal in terms of overall download time. In Section \ref{sec:equilib}, we describe the equilibria under continuous globally (and locally) rarest first control. In Section \ref{sec:controlled-rarity}, we discuss how the rarity of file segments can be  deliberately controlled by the seeders. In Section 
\ref{rare-uplink-sec}, we interpret a two-segment model as one modeling
a rare segment and a lumping of the rest to study rare-segment control under 
choking due to differences in uplink bandwidths of the peers.
Finally, we conclude in Section \ref{concl-sec} with a summary.

\section{Epidemic model}\label{sec:EpidemicModel}
Let $\lambda$ be the total peer arrival rate to a specific swarm, where newly arrived peers possess no part of the data object $F$ being disseminated in the swarm. The quantity $\delta$ is the death rate for seeders--individuals that possess the entire file and seed the population with its segments. Let $\beta$ be the download rate parameter of client-server transactions, which depend on the size of the file being transmitted and the associated willingness of the server peer to participate in the transaction (for nothing in return from the client peer). 

\subsection{One Segment Model}
In the absence of BitTorrent incentives, we have the single-segment case
\begin{displaymath}
\dot x_l = - \beta_0 x_l x_s + \lambda_l \quad
\dot x_s = \beta_0 x_l x_s -\delta x_s+\lambda_s,
\end{displaymath}
where $x_l$ are the leechers (and do not possess any parts of the file),
$x_s$ are the seeders (who possess the complete file), and
$\lambda$ their exogenous arrival rates to the swarm.
 The successful transaction rate is  proportional to the contact rate between a member of the seeder and leecher populations, which we assume to be proportional to  the product of their sizes \cite{Daley-Gani}. We chose this model for concreteness; other types of models, some similar to  the above, have also been extensively studied, e.g., urn, replicator, Volterra-Lotka, and coupon-collector \cite{Voj-Mass05}.

Therefore there are  two types of peers in the seeder state: those that arrive as seeders and tend to remain longer  and those that arrive formerly as leechers and tend to remain briefly. Rather than using a fixed population of ``permanent" seeders for the former category, we model them by a small external arrival process  $\lambda_s$, giving an average population of $\lambda_s/\delta_s$ by Little's formula \cite{Wolff89}, with $\frac{1}{\delta_s}\gg \frac{1}{\delta_l}$. The mean lifetime $1/\delta$  of a {\em typical} seeder is therefore the weighted average:
\begin{equation}
\frac{1}{\delta} = \frac{\lambda_l/\delta_l}{\lambda_l/\delta_l+\lambda_s/\delta_s}\cdot\frac{1}{\delta_l}
+ \frac{\lambda_s/\delta_s}{\lambda_l/\delta_l+\lambda_s/\delta_s}\cdot\frac{1}{
\delta_s}
\end{equation}

The globally attracting stable equilibrium is given by 
\begin{displaymath}
\mathbf{x}^* = (x_l^*, x_s^*) =  \left(
\frac{\delta\lambda_l}{\beta(\lambda_s+\lambda_l)}, ~ 
\frac{\lambda_l+\lambda_s}{\delta}\right),
\end{displaymath}

\subsection{Two-Segment Model}

Consider splitting file $F$ into two segments, $a$ and $b$. In this case, the model for this system becomes:
\begin{equation}
\left\{
\begin{aligned}
\dot x_l & = \lambda_l - \beta  x_l ( x_a+ x_b+ x_s) \\
\dot {x}_a &  = - x_a ( \beta  x_s +\gamma  x_b) +  
	\beta  x_l (  x_a + u(x_a,x_b) x_s )\\
\dot {x}_b & = - x_b ( \beta  x_s   + \gamma  x_a) +  
	\beta  x_l (  x_b + [1-u(x_a,x_b)] x_s   )\\
\dot{x}_s & =  \lambda_s + \beta ( x_a+ x_b )  x_s   + 
	2 \gamma  x_a  x_b -\delta  x_s
\end{aligned}\right.
\label{eqn:SystemModel}
\end{equation}
Here, $\gamma$ represents is  rate parameter  for swap/trade transactions ($\beta$ and $\gamma$ may be decreasing functions $N$). The ``control" function $u \in [0,1]$ represents how the seeder may distribute the segments based on its knowledge of their relative prevalence; whereas in BitTorrent, the rarity of the segment is locally determined  among peers that are directly transacting, i.e., BitTorrent uses \textit{locally} rarest first policy to determine which segments to disseminate. In \cite{NETCOOP08}, $u \equiv 1/2$ was assumed.

We assume that $\gamma \geq \beta > \beta_0$,
where the former inequality is owing to stronger ``server" incentives in a swap transaction compared to a client server transaction.  Also, the latter inequality is owing  to $a$ and $b$ being smaller than $F$ and peers would generally be more reluctant to transmit the entire file $F$ in one shot out of concerns of liability for copyright violation. In \cite{NETCOOP08}, for $u \equiv 1/2 $ (constant control), we showed convergence of a scaled stochastic discrete transactional process to the above epidemic dynamics  and compared pure client-server to the two-segment system in terms of time to transition from leecher to seeder.

It is worth noting that BitTorrent does not function precisely in this way. Permanent BitTorrent seeders do not truly have control over the pieces they chose to transmit to members of the swarm, as swarms use pull (as opposed to push) request frames. This model assumes that seeders (through some mechanism) will have control over the fragments they push to seeders through the $u(x_a,x_b)$ parameter.

\section{Two-Segment Selection Control}\label{sec:SelectionControl}

For simplicity in the following,
we focus on the two-segment swarm (\ref{eqn:SystemModel}).
When $u\equiv 1/2$, then half the successful transactions between $x_l$ and $x_s$ result in an arrival to the $ x_a$ population (that possess only the first segment of $F$),  and the other half an arrival to the $x_b$ population (in both cases, a departure from the $ x_a$ population,
of course). 
When $u \equiv 1/2$, 
there is always at least one equilibrium solution 
given by
\cite{NETCOOP08}:
\begin{equation}
\begin{aligned}
x_l = &\frac{\lambda_l}{\beta}\left(\sigma_0+\frac{\lambda_l+\lambda_s}{\delta}\right)^{-1} &\quad 
x_a = & \frac{\sigma_0}{2},\\
x_b = & \frac{\sigma_0}{2}, &\quad
x_s = & \frac{\lambda_l+\lambda_s}{\delta} 
\end{aligned}
\label{eqn:u_equiv_half}
\end{equation}
where $\sigma_0$ is the unique positive root of the quadratic equation:
$\sigma_0^2 + 2\kappa_0 \sigma_0 - {2\lambda_l}/{\gamma}$
and
$\kappa_0  := \frac{ \beta(\lambda_l+\lambda_s)}{\gamma \delta}$.

The case of a \textit{constant control}  (allowing for a constant $u \in [0,1]$) is complex and leads to the analysis of a quartic equation without providing much insight into the system. We consider this case in Section \ref{rare-uplink-sec} through a numerical study. In general, the control will vary as a function of state $(x_l,x_a,x_b,x_s)$. The \textit{continuous globally rarest first} control is simply
\begin{equation}
u(x_a,x_b) =  \frac{x_b}{x_a+x_b}
\label{eqn:lrf-cts}
\end{equation}
The presumption here is that the seeders have an estimate of the ratio of population sizes $x_a/x_b$. A non-continuous, ``bang-bang" version of this rule, requiring less information for the seeders, is
\begin{equation}
u(x_a,x_b) = \begin{cases}
1 & \text{if $x_a < x_b$}\\
\frac{1}{2} & \text{if $x_a = x_b$}\\
0 & \text{if $x_a > x_b$}
\end{cases}
\label{eqn:lrf-bang-bang}
\end{equation}
Again, BitTorrent used a \textit{locally} rarest first control consistent with (\ref{eqn:lrf-bang-bang}). Note that both of these controls admit the equilibrium for $u \equiv 1/2$ of the previous section.

\section{Minimizing Traversal Time}\label{sec:opt-control-sec}
Consider the Mayer optimal control problem, with control $u(\mathbf{x},t)$:
\begin{equation}
\left\{
\begin{aligned}
\min_u\;\; & x_l(T)+x_a(T)+x_b(T) ~\text{subject to:}\\
&\text{the system model (\ref{eqn:SystemModel}),}\\
&\mathbf{x}(0) = \mathbf{x}^0,\\
&u \in [0,1],\, \mathbf{x}(t) \geq 0\;\;\forall t \in [0,T].
\end{aligned}
\right.
\label{eqn:OptControl}
\end{equation}
Here we assume that $T$ is a finite ending time that may be arbitrarily large. Naturally we assume that $\mathbf{x}, u\in \mathcal{L}^2[0,t]$, the space of bounded square integrable functions and $T$ is some arbitrary large but \textit{finite} end of time.

The objective is motivated by Little's formula \cite{Wolff89} which states that, {\em for a stationary regime}, the sojourn time 
from arrival to a swarm as leecher to the transition to seeder
  is 
\begin{equation}
\frac{x_l^* + x^*_a + x^*_b}{\lambda_l}.
\label{Little-sojourn}
\end{equation}
 Our assertion that $T$ is finite comes from the qualitative analysis of the differential equations given in System (\ref{eqn:SystemModel}). We argue that Expression (\ref{eqn:lrf-bang-bang}) is the control that minimizes the objective subject to these epidemic dynamics, beginning from an arbitrary initial point $\mathbf{x}(0)$. Expression (\ref{eqn:lrf-bang-bang}) is the fully discrete form of Expression (\ref{eqn:lrf-cts}). 

There is always at least one attracting equilibrium point for the epidemic dynamics (\ref{eqn:SystemModel}) when $u \equiv 1$ (this is the value of Expression (\ref{eqn:lrf-bang-bang}) when $x_a < x_b$). This equilibrium occurs at:
\begin{displaymath}
x^*_a = \frac{\lambda\delta}{(\lambda_s + \lambda)\beta},  \quad 
x^*_b = 0.
\end{displaymath}
If we assume that $u$ is defined by Expression (\ref{eqn:lrf-bang-bang}) and that $x_a(0) < x_b(0)$, we note that before the foregoing equilibrium is reached, we obtain $x_a(t) = x_b(t)$ for some $t$, and we return to the dynamics in the case when $u \equiv 1/2$ (see Figure \ref{fig:Control}). In this case, we will maintain $x_a(t) = x_b(t)$ and move to the equilibrium point already identified in Expression (\ref{eqn:u_equiv_half}). A similar argument holds when $x_b(0) < x_a(0)$ in which case $u \equiv 0$. Again we will return to the dynamics when $u \equiv 1/2$ before $x_a$ reaches $0$, which is the equilibrium in this case. 

\begin{figure}[htbp]
\centering
\includegraphics[scale=0.4]{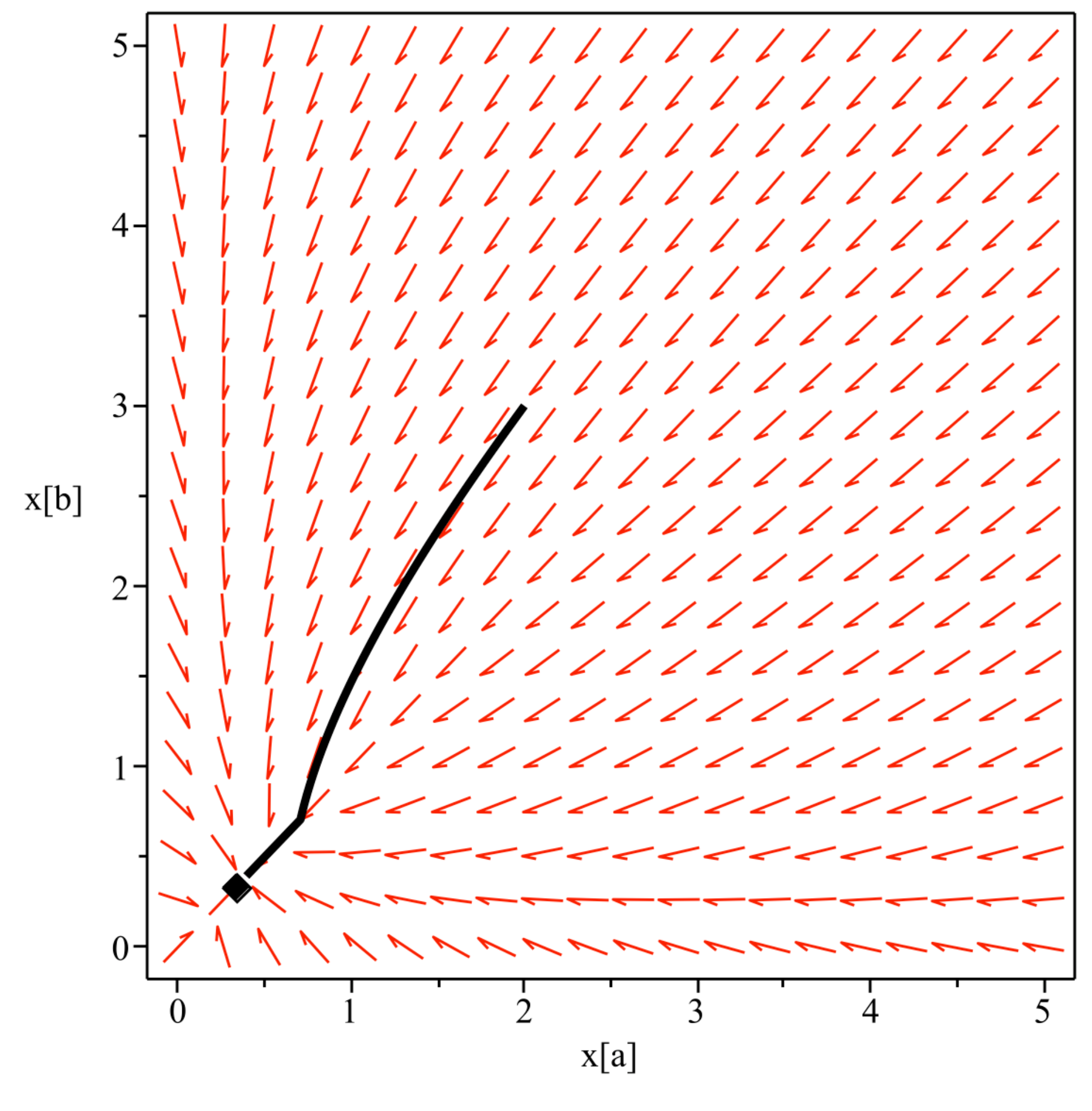}
\caption{Phase plot in $(x_a, x_b)$ space showing the bang-bang controller pushing $x_a$ to equal $x_b$ and then proceeding to a globally attracting stationary point. In this example, $\beta = 2$, $\gamma = 3$, $\lambda_s = 1$, $\lambda_l = 4$ and $\delta = 2$.} 
\label{fig:Control}
\end{figure}

Observe that the Hamiltonian of the control problem 
(\ref{eqn:OptControl})
will be linear in $u$. Thus,  the bang-bang controller
of Expression (\ref{eqn:lrf-bang-bang}) is optimal
 \cite{bang-bang-ref}.
 The bang-bang controller will switch its state depending on the adjoint dynamics of the system and may be singular for certain adjoint dynamics. The only reasonable singular control in this case is $u \equiv 1/2$ 
(used whenever $x_a=x_b$), while a reasonable proxy for the adjoint conditions is given in Expression (\ref{eqn:lrf-bang-bang}). 

We illustrate the optimality of the discontinuous globally rarest first controller through a numerical example. Figure \ref{fig:Controller} at left shows the optimal controller, which was computed by discretization, pushing $x_a = x_b$ and then maintaining this state.

\begin{figure}[htbp]
\centering
\subfigure[Bang-Bang Control]{\includegraphics[scale=0.4]{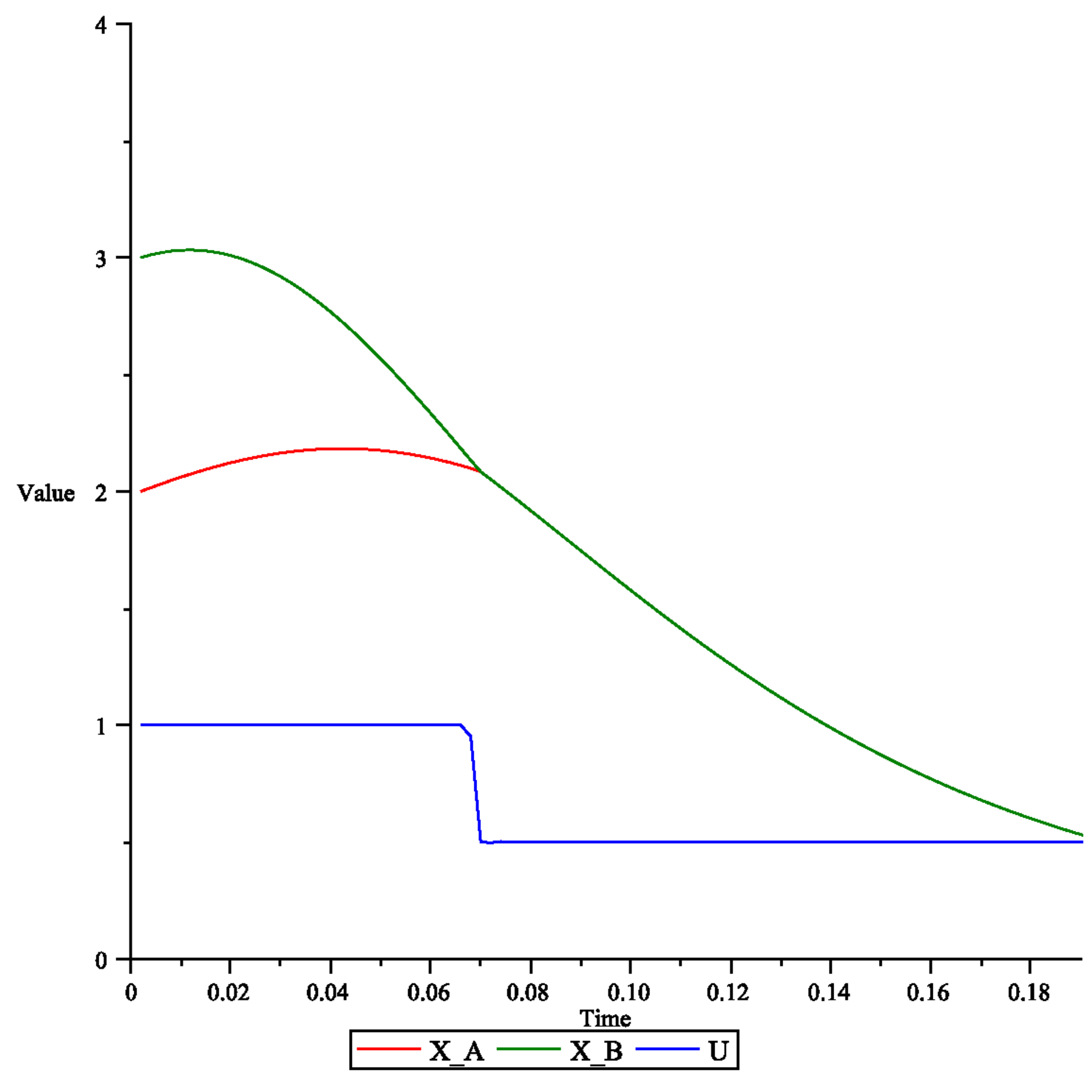}}
\subfigure[Continuous Control]{\includegraphics[scale=0.4]{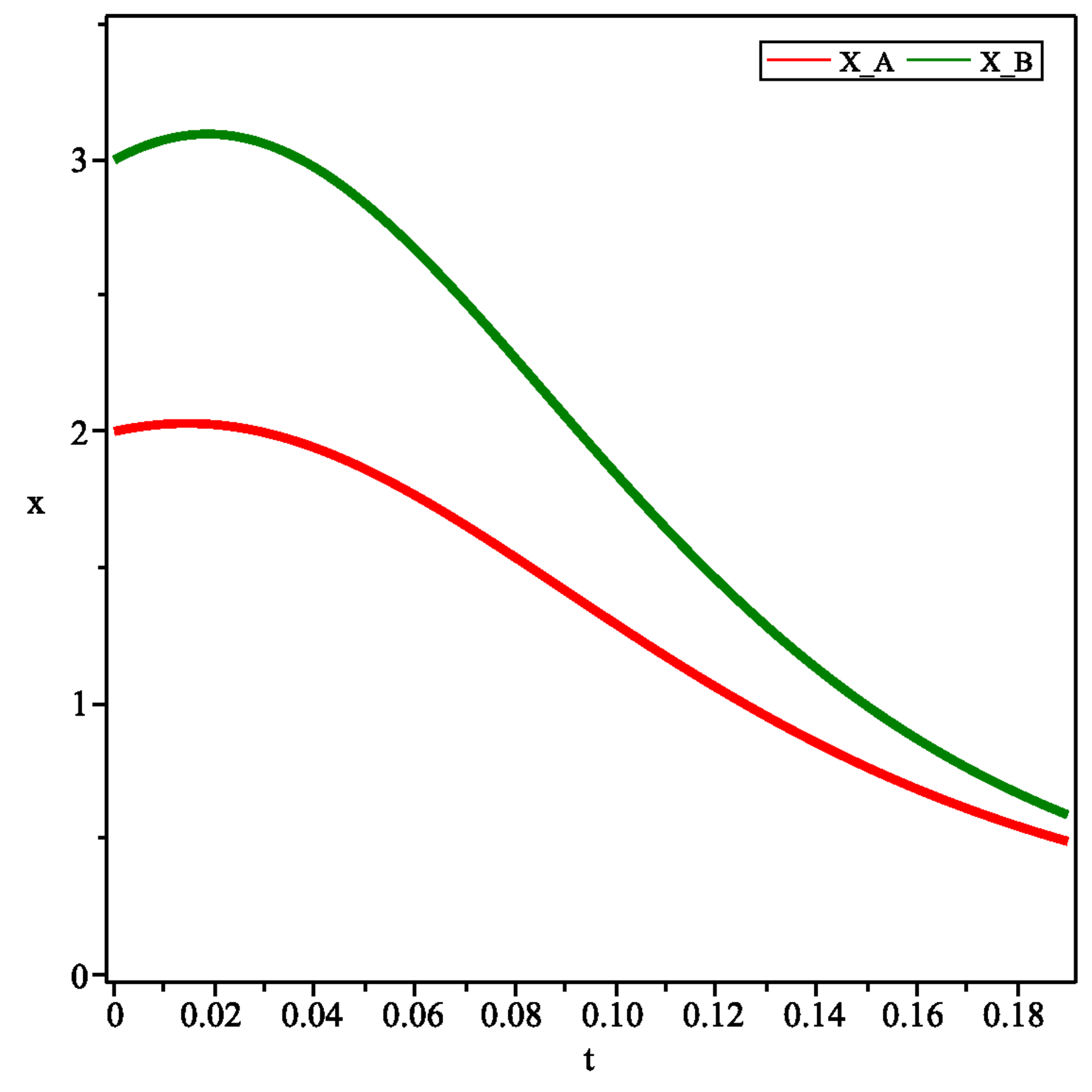}}
\caption{Left: The optimal controller is shown on the bottom, while $x_a(t)$ and $x_b(t)$ are shown above. Note the optimal controller driving $x_a = x_b$. This controller was computed using discretization method. Right: The continuous globally rarest first function drives $x_a$ to to $x_b$ however, the convergence rate is slower than for the true optimal control function.}
\label{fig:Controller}
\end{figure}

We can contrast this to the continuous approximation of the globally rarest first controller given in Expression (\ref{eqn:lrf-cts}). In Figure \ref{fig:Controller} at right we illustrate the effect the continuous globally rarest first controller has on the values of $x_a$ and $x_b$. Note that the two converge much more slowly than in the optimal case.

Note that since Problem (\ref{eqn:OptControl}) is of the Mayer type,
the specific form of the controller will matter most when $T$ is small, 
i.e., when $T$ is much less than the time required to 
reach an attracting equilibrium. 
This is precisely the behavior we see in Figures \ref{fig:Controller}. 
The optimal (bang-bang) controller attempts to drive the system 
to its equilibrium point as quickly as possible since it is here that 
steady-state  component of
Problem (\ref{eqn:OptControl}), given 
by  (\ref{Little-sojourn}),
is minimized. 
It does so by pushing the system to a spot on the diagonal 
(where $x_a = x_b$); since thereafter $u \equiv 1/2$,
the systems with constant control
($u\equiv 1/2$ always), 
discrete globally rarest first control, or continuous globally rarest first 
control, will share equilibrium points 
(this is illustrated in the next section). 
For large values of $T$, the dominant component of the
objective   will be the 
the equilibrium point to which the system tends, not the transient
component  which may be sensitive to the  choice among such controllers.

\section{Equilibria under Continuous Globally Rarest First Control}
\label{sec:equilib}
More interesting equilibria are possible under the continuous form of globally rarest first controller (\ref{eqn:lrf-cts}). Equilibrium analysis focusing on the values of $x_a^*$ and $x_b^*$ (by first solving for and 
substituting out $x_s^*$ and $x_l^*$) yields:
\begin{equation}
x_l^* = \lambda_l{\beta}^{-1} \left( x^*_{{a}}+x^*_{{b}}-{\frac {\lambda_s+2\,\gamma\,x^*_{
{a}}x^*_{{b}}}{\beta\,x^*_{{a}}+\beta\,x^*_{{b}}-\delta}} \right) ^{-1}.
\label{eqn:XLSOLN}
\end{equation}
\begin{equation}
x_s = -{\frac {\lambda_s+2\,\gamma\,x_{{a}}x_{{b}}}{\beta\,x_{{a}}+\beta\,x_{{b}}
-\delta}}
\label{eqn:XSSOLN}
\end{equation}
In the following, 
suppress the superscript ``*" for notational simplicity.

We can solve the simpler simultaneous nonlinear equations
$\dot{x}_a + \dot{x}_b  = 0$ and
$\dot{x}_a - \dot{x}_b  = 0$
to obtain equilibrium solutions for $x_a$ and $x_b$ respectively. From the former, we get:
\begin{equation}
x_a = -{\frac {x_{{b}}\lambda_l\,\beta+\beta\,\lambda_s\,x_{{b}}-\lambda_l\,\delta}{2
\,\gamma\,\delta\,x_{{b}}+\lambda_l\,\beta+\beta\,\lambda_s}}
\label{eqn:XASOLN}
\end{equation}
We can then substitute this expression into $\dot{x}_a - \dot{x}_b=0$ 
to obtain a ratio of polynomials in $x_b$ whose numerator is:
\begin{equation}
-2\left(2\gamma\delta x_b^{2}+ 2\beta(\lambda_l + \lambda_s)x_b-\lambda_l\delta\right)\left[a_0 + a_1x_b + a_2x_b^2 + a_3x_b^3 + a_4x_b^4\right]
\end{equation}
where:
\begin{displaymath}
\left\{
\begin{aligned}
&a_0 = {\lambda_l}^{4}{\beta}^{2}+{\lambda_s}^{3}{\beta}^{2}\lambda_l+3\,\lambda_s\,{\beta}
^{2}{\lambda_l}^{3}+3\,{\lambda_s}^{2}{\beta}^{2}{\lambda_l}^{2}\\
&a_1 = 3\,{\lambda_s}^{2}\lambda_l\,\delta\,\gamma\,\beta-{\lambda_l}^{2}\gamma\,{
\delta}^{3}+6\,\lambda_s\,{\lambda_l}^{2}\delta\,\gamma\,\beta+3\,\delta\,
\gamma\,\beta\,{\lambda_l}^{3}\\
&a_2 = 3\,{\beta}^{2}\lambda_s\,{\lambda_l}^{2}\gamma+3\,{\beta}^{2}{\lambda_s}^{2}
\lambda_l\,\gamma+\gamma\,{\delta}^{2}\lambda_l\,\beta\,\lambda_s+2\,{\gamma}^{
2}\lambda_l\,{\delta}^{2}\lambda_s\\
& \hspace*{3em}+2\,{\gamma}^{2}{\lambda_l}^{2}{\delta}^{2}+
\gamma\,\beta\,{\delta}^{2}{\lambda_l}^{2}+{\beta}^{2}\gamma\,{\lambda_s}^{3}
+{\beta}^{2}\gamma\,{\lambda_l}^{3}\\
&a_3 = 4\,\lambda_s\,\lambda_l\,\delta\,{\gamma}^{2}\beta+2\,{\lambda_s}^{2}\delta\,{
\gamma}^{2}\beta+2\,{\lambda_l}^{2}\delta\,{\gamma}^{2}\beta-2\,{\gamma}
^{2}{\delta}^{3}\lambda_l\\
&a_4 = 2\,{\gamma}^{2}{\delta}^{2}\lambda_l\,\beta+2\,{\gamma}^{2}{\delta}^{2}\beta\,\lambda_s
\end{aligned}\right.
\end{displaymath}
The roots of 
$\dot{x}_a - \dot{x}_b$ are governed by the roots of the quadratic polynomial:
\begin{equation}
2\gamma\delta x_b^{2}+ 2\beta(\lambda_l + \lambda_s)x_b-\lambda_l\delta
\label{eqn:quad}
\end{equation}
and a quartic polynomial:
\begin{equation}
a_0 + a_1x_b + a_2x_b^2 + a_3x_b^3 + a_4x_b^4
\label{eqn:quart}
\end{equation}

\subsection{Quadratic Equation Case}\label{quad-sec}
\begin{proposition} 
A unique strictly positive, real solution with $x_a=x_b$ exists for
(\ref{eqn:XASOLN}) and (\ref{eqn:quad}).
\label{claim:UniqueQuad}
\end{proposition}
\begin{proof} The roots of (\ref{eqn:quad}) are given by:
\begin{equation}
\frac{-1}{4\gamma\delta}\left(2\beta(\lambda + \rho) \pm 
2 \sqrt{\beta^2(\lambda + \rho)^2 + 2\lambda\gamma\delta^2} \right)
\end{equation}
The fact that all parameters are positive yields the expression:
\begin{equation}
2\beta(\lambda + \rho) = 2\lvert\sqrt{\beta^2(\lambda + \rho)^2}\rvert < 
2\lvert \sqrt{\beta^2(\lambda + \rho)^2 + 2\lambda\gamma\delta^2}\rvert
\end{equation}
Thus, 
\begin{equation}
2\beta(\lambda + \rho) - 
2 \sqrt{\beta^2(\lambda + \rho)^2 + 2\lambda\gamma\delta^2} < 0
\end{equation}
while 
\begin{equation}
2\beta(\lambda + \rho) + 
2 \sqrt{\beta^2(\lambda + \rho)^2 + 2\lambda\gamma\delta^2} > 0
\end{equation}
The factor of $-1/4\gamma\delta$ leads us to conclude that there is one positive root and one negative root always for (\ref{eqn:quad}). Thus:
\begin{equation}
x_b^* = \frac{-1}{4\gamma\delta}\left(2\beta(\lambda + \rho) - 
2 \sqrt{\beta^2(\lambda + \rho)^2 + 2\lambda\gamma\delta^2} \right)
\end{equation}
is a non-extraneous equilibrium solution. Substituting this value in the expression for $x_a$ and simplifying algebraically yields:
\begin{equation}
x_a^* = \frac{-1}{4\gamma\delta}\left(2\beta(\lambda + \rho) - 
2 \sqrt{\beta^2(\lambda + \rho)^2 + 2\lambda\gamma\delta^2} \right)
\end{equation}
as well. That is, the two roots are equal and positive. This completes the proof.
\end{proof}

\subsection{Quartic Equation Case}
The roots of the quadratic equation (\ref{eqn:quad}) just considered are not necessarily those of the quartic equation (\ref{eqn:quart}). We make use of the following known
result (see, e.g.,  Theorem 1 of \cite{quartic-roots-ref}): 
\begin{theorem} For the quartic equation,
$f(x) = c_0 + 4c_1x + 6c_2x^2 + 4c_3x^3 + c_4x^4$,
define the following terms:
$G = c_4^2c_1-3c_4c_3c_2+2c_3^3$,
$H = c_4c_2-c_3^2$,
$I = c_4c_0-4c_3c_1+3c_2^2$,
\begin{displaymath}
J = \left\lvert
\begin{array}{ccc}
c_4 & c_3 & c_2\\
c_3 & c_2 & c_1\\
c_2 & c_3 & c_4
\end{array}
\right\rvert,
\end{displaymath}
and the discriminant $\Delta = I^3 - 27J^2$.
Then $f(x) = 0$ has no real roots if and only if:
\begin{enumerate}
\item $\Delta = 0$, $G=0$, $12H^2-c_4I=0$ and $H > 0$; or
\item $\Delta > 0$ and
 (a) $H\geq 0$; or
 (b) $H < 0$ and $12H^2 - c_4^2I < 0$
\end{enumerate}\hfill\qquad\endproof
\label{thm:Roots}
\end{theorem}

For (\ref{eqn:quart}) we can evaluate the discriminant and determine conditions on $\lambda_s$ when $\Delta < 0$, which will create a real-root for (\ref{eqn:quart}). The sign of the discriminant is governed by a quadratic expression on $\lambda_s$. We identify terms:
\begin{equation}
\lambda_{0,1}  =  
\frac {(\eta\xi+1)^2}{4\gamma\,{\xi}^{3}\eta}\left[
 10\,\beta\,\gamma+{\gamma}^{2} ~ \pm\sqrt {68\,{\beta}
^{2}{\gamma}^{2}+20\,\beta\,{\gamma}^{3}+{\gamma}^{4}+64\,\gamma\,{
\beta}^{3}}\right]
\label{lambda_k-def}
\end{equation}
Given the fact that the parameters are always positive, the square root is always real and thus for any set of parameters, we can identify a condition under which $\Delta = 0$. Let $\lambda_0$ and $\lambda_1$ denote the two roots defined above. Evaluating a point directly between the two roots yields:
\begin{equation}
\lambda_s^{+} = {\frac { \left( \eta\,\xi+1 \right) ^{2} \left( \gamma+10\,\beta
 \right) }{4{\xi}^{3}\eta}}
\end{equation}
Evaluating $\Delta$ at $\lambda_s^+$ yields the simple expression:
\begin{equation}
\Delta(\lambda_s^+) = \frac{1}{8}\left( \eta\,\xi+1 \right) ^{4} \left( \gamma+16\,\beta \right) 
 \left( \gamma+2\,\beta \right) ^{2}
\end{equation}

Since all parameters are positive, we can see that when $\lambda_s \in (\lambda_0,\lambda_1)$, then $\Delta > 0$. If we choose $\lambda_s$ outside of this range and fix $\beta$, $\gamma$, $\eta$ and $\xi$ we can identify the example off-diagonal equilibrium solutions illustrated in the main text.

For the case:   
\begin{displaymath}
\lambda_l \geq \delta \geq \lambda_s,
\end{displaymath}
let $\eta=\lambda_l/\delta$ and $\xi=\delta/\lambda_s$ with $\eta, \xi\geq 1$. The quartic discriminant $\Delta$ will depend on a quadratic form in $\lambda_s$ with roots at $\lambda_0$ and $\lambda_1$ and $\lambda_0<\lambda_1$. Between these values, the quartic equation has no real roots. Outside these values, the quartic has at least two real roots which correspond to equilibrium points for $x_a$ and $x_b$ that are off-diagonal (i.e., $x_a \neq x_b$).

\subsection{Off-Diagonal Equilibria}
\label{sec:off-diag}
Real solutions to the quartic equation are interesting because they lead to off-diagonal equilibrium solutions in cases where the values of the parameters are widely skewed. For the case $\beta = 2$, $\gamma = 3$, $\eta = 1.1$, $\xi = 1.1$ we obtain:
\begin{equation}
\lambda_0 = -0.759  \text{ and } \lambda_1 = 39.121.
\end{equation}
Choosing $\lambda_s = 40>\lambda_1$, so that $\delta = 44$ and $\lambda_l = 48.4$. The resulting (real) off-diagonal equilibrium points $(x_a^*,x_b^*)$ derived with these parameters in the quartic are:
\begin{displaymath}
(5.237,~0.772) \text{ and } (0.772,~5.237).
\end{displaymath}
This example illustrates the existence of off-diagonal equilibrium points. A field plot of this case is shown in Figure \ref{fig:OffDiagonal} (Left). The blue lines are the trajectories of the dynamical system with representative starting points. Off-diagonal equilibrium points are shown as black diamonds.

It is interesting to note that the off-diagonal equilibrium points partition the phase plane into a  central region of stability flanked by two regions of instability \cite{Packard10}. Thus, when parameter values are highly skewed, the on-diagonal equilibrium point is not a global attractor as it appears to be when it is the unique non-extraneous equilibrium (see Figure \ref{fig:OffDiagonal} on Left). The trajectories illustrate a component of this region of attraction. It is clear that these equilibrium points \textit{should not} occur in the bang-bang control case -- evaluating Little's formula for the on-diagonal equilibrium point in this case shows that the on-diagonal equilibrium correctly minimizes the objective function of the control problem (\ref{eqn:OptControl}). This is discussed in the next section.

The results shown on the equilibrium points of the continuous locally rarest first control are summarized in the following theorem.

\begin{theorem} For the dynamics given in Expression (\ref{eqn:SystemModel})
under (\ref{eqn:lrf-cts}), 
there is alway at least one point of equilibrium occurring at:
\begin{gather*}
x_s^* = \frac{\lambda_l + \lambda_s}{\delta}\\
x_l^* = \frac{\lambda_l\gamma\delta}
{\beta\left((\lambda_l + \lambda_s)(\gamma-\beta) + \sqrt{\beta^2(\lambda_l+\lambda_s)^2 + 2\lambda_l\gamma\delta^2}\right)}\\
x_a^* = x_b^* = \frac{-1}{4\gamma\delta}\left(2\beta(\lambda_l + \lambda_s) - 
2 \sqrt{\beta^2(\lambda_l + \lambda_s)^2 + 2\lambda_l\gamma\delta^2} \right)
\end{gather*} 
Furthermore, if $\lambda_s \not \in [\lambda_0,\lambda_1]$ and $\lambda_s>0$, with $\lambda_0$ and $\lambda_1$ defined in Expression (\ref{lambda_k-def}), then the system may admit at least one other equilibrium point; in this case, $x_l^*$ and $x_s^*$ remain as defined, but $x_a^*$ and $x_b^*$ may take on non-equal values. 
\end{theorem}

We illustrate the on diagonal equilibrium that always exists in Figure \ref{fig:OffDiagonal} (Left and Right)  for the case when $\beta = 2$, $\gamma = 3$, $\lambda_s = 1$, $\lambda_l = 4$ and $\delta = 2$. The black line shows a representative sample path for this dynamical system. 
\begin{figure}[htbp]
\centering
\subfigure[Off-Diagonal Equilibrium]{\includegraphics[scale=0.4]{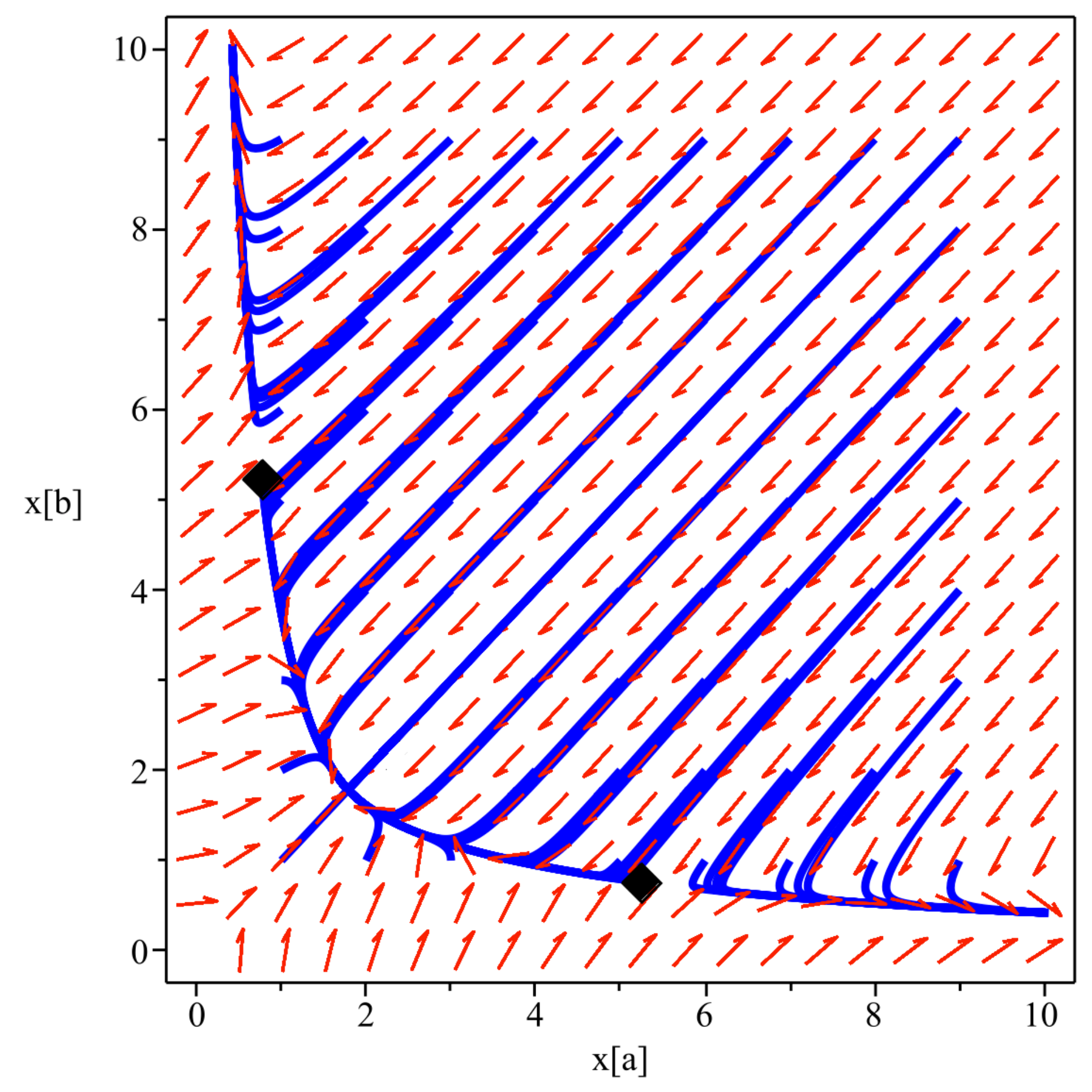}}
\subfigure[On Diagonal Equilibrium]{\includegraphics[scale=0.4]{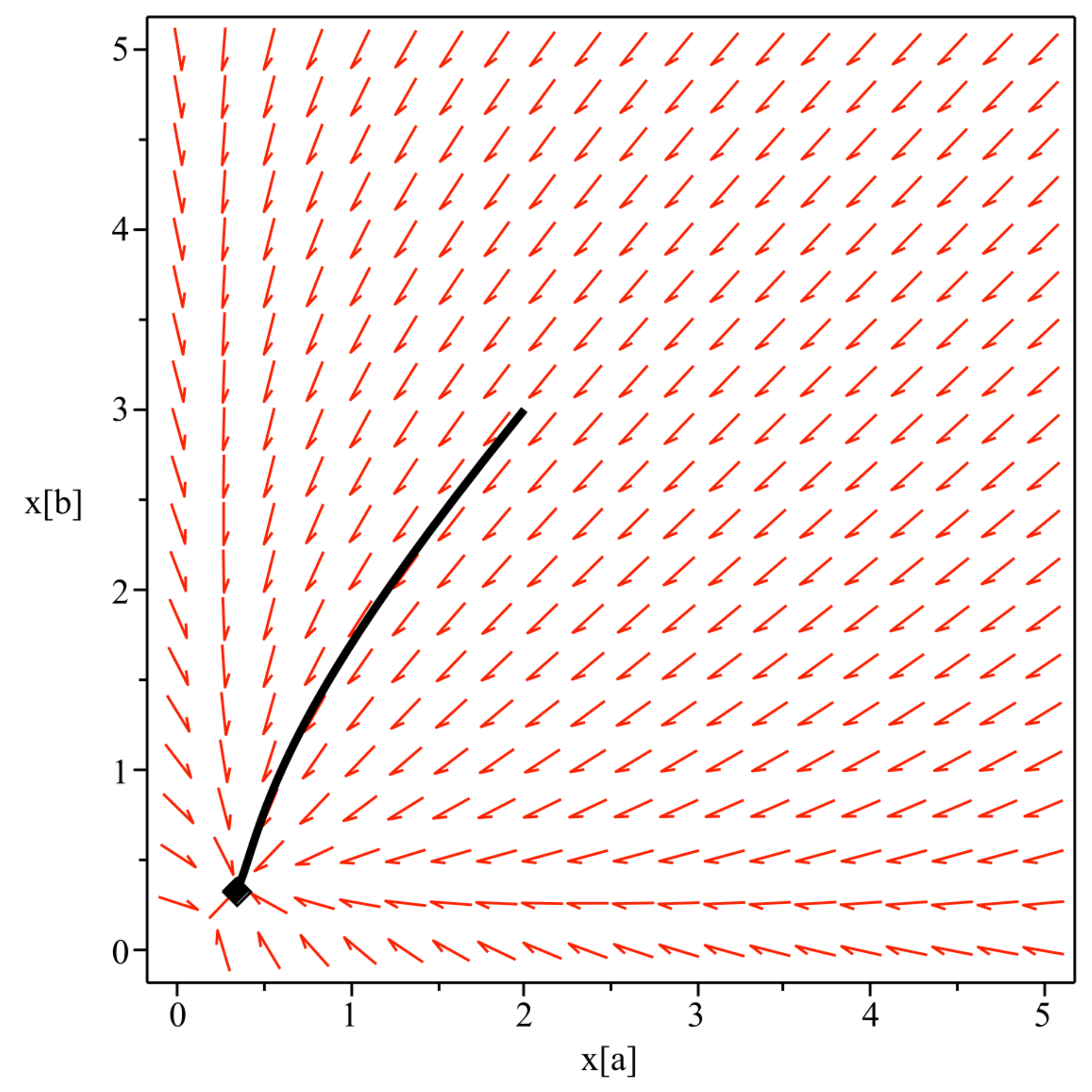}}
\caption{(Left) Phase plot in $(x_a, x_b)=(x,y)$ space showing regions of stability and instability with off-diagonal equilibrium points. (Right) Phase plot in $(x_a, x_b)=(x,y)$ space showing the stability of the on-diagonal equilibrium point. The black line shows a representative sample path for this dynamical system.}
\label{fig:OffDiagonal}
\end{figure}

\section{Discussion: Controlling segment rarity}
\label{sec:controlled-rarity}
Controlled rarity is the action of keeping a collectible object from a set intentionally rare. 
Controlled rarity is used to sell packs of trading cards, i.e.,  collectors buy additional packs seeking the rare card to complete their set. In BitTorrent, the rarity of certain segments can be deliberately controlled to increase swarm sojourn time and encourage additional cooperative (uploading) behavior by leechers. This can serve to stabilize the swarm and prevent collapse.

As a simple example in the two-segment case,  the seeders can use $u(x_a,x_b)  = x_a/(x_a + k x_b)$ for some constant $k$. Clearly, as  $k>1$  increases, the seeders will increasingly tend to disseminate segment $b$ even when  segment $a$ is rarer in the swarm. Based on the results of the previous section, the sojourn time from  leecher to seeder is larger under any such rule with compared to the globally rarest first rule (\ref{eqn:lrf-bang-bang}). At an extreme, seeders could substitute $1-u$ for $u$ 
given in (\ref{eqn:lrf-bang-bang}) to extend leecher sojourn times.

Heretofore, we have described three control policies, (\ref{eqn:lrf-cts}), (\ref{eqn:lrf-bang-bang}), and $u\equiv 1/2$. We note that these three control policies share equilibrium points with $x_a^*=x_b^*$. But for (\ref{eqn:lrf-cts}), we computed off-diagonal equilibria  with $x_a^*+x_b^* = 5.237+0.772= 6.009$. However, for this example, the on-diagonal equilibria are $x_a^*=x_b^* = 2.248$, i.e., $x_a^*+x_b^* = 4.496$  which is less than that of the off-diagonal equilibria. So, in this way, we see that, even in  a stationary regime, control (\ref{eqn:lrf-cts}) may still lead to longer leecher sojourn times (through Little's Formula) than the optimal bang-bang control (\ref{eqn:lrf-bang-bang}). 

Our model allows us to devise beneficial delaying strategies in terms of reduced load for the permanent seeders, overall content availability, and performance--at least for a percentage of the leechers, since clearly those that will be delayed will not have any gain for the specific swarm. We leave this for future work.

\section{Jointly modeling uplink-based choking and rare segments with a lumped non-rare segment model}\label{rare-uplink-sec}

In this section, we reinterpret the two-segment model as representing a general BitTorrent swarm with one segment intentionally rare. That is, all other (non-rare) segments are lumped together by assuming that their collective acquisition
occurs at the same time-scale as that of the single rare segment. 
Thus, the four types of players based on the segments they possess are: 
leechers, seeders, players with all but the rarest segment 
(i.e., $N-1$ segments), and players with the rare segment. 

Let $\beta_{k}$ ($k\in \{r,N-1\}$) be the probability of successful transaction,
again assuming contact/attempts are proportional to the probability of the populations. The probability $\beta_r$ for transactions involving the rare segment are such that $\beta_r \leq \beta_{N-1}$, where
$\beta_{N-1}$ is the lumped parameter corresponding to  peers
with all $N-1$ other (non-rare) segments.

As a result, we get the following epidemic dynamics as a variation of those considered above assuming all uplinks are the same:
\begin{align}
\dot{x}_l  = & \lambda_l -(\beta_r x_r +\beta_{N-1}x_{N-1}
+ [u\beta_r + (1-u)\beta_{N-1}]x_s) x_l\\
\dot{x}_r  = &  \beta_r (u x_s +x_r)x_l -(\beta_{N-1}x_s 
+\beta_r x_{N-1})x_r\\
\dot{x}_{N-1}  = &  \beta_{N-1} ((1-u) x_s +x_{N-1})x_l 
-\beta_r(x_s + x_r)x_{N-1}\\
\dot{x}_s  = & \lambda_s +\beta_{N-1}x_s x_r +\beta_r(2x_r+x_s)x_{N-1}
-\delta x_s 
\end{align}
where we require $\lambda_s,\lambda_l > 0$ to prevent extinction. Also, consider how the seeder can control the system by varying $u$ governing seeder contact with leechers (whether a rare segment is
shared).  For a simple example,
take $\beta_r=1=\beta_{N-1}$, $\lambda_l=1$, $\lambda_s=0.01$ and $\delta=0.1$. In Figure \ref{no-choke-fig}, we plot as a function of $u$
the mean (i.e., equilibrium, $\dot{\underline{x}}=0$) delay from leecher to seeder, which is by Little's theorem \cite{Wolff89}
\begin{eqnarray*}
\frac{x_l + x_r + x_{N-1}}{\lambda_l}.
\end{eqnarray*}
The equilibrium point for these ``symmetric"  dynamics,
involving the roots of (\ref{eqn:quad}), can be computed as follows:
For the symmetric case where $\beta_r=1=\beta_{N-1}$ and $u=0.5$, the diagonal equilibrium solution is
\begin{align}
x_l^* = & 
\frac{x_{N-1}^*\lambda_l\delta}
{\lambda_l \delta - x_{N-1}^*(\lambda_s+\lambda_l) },\\
x_r^*,x_{N-1}^* = & 
\frac{-(\lambda_s+\lambda_l)
+\sqrt{(\lambda_s+\lambda_l)^2+ 2\delta^2\lambda_l} }
{2\delta},\\
x_s^* = & \frac{\lambda_s+\lambda_l}{\delta},
\end{align}
where 
$x_{N-1}^*=x_r^*$ is the positive root  (\ref{eqn:quad}).

Other solutions (when $u \neq 1/2$) are computed as the stationary points of the differential system.
\begin{figure}[htbp]
\centering
\includegraphics[scale=0.4]{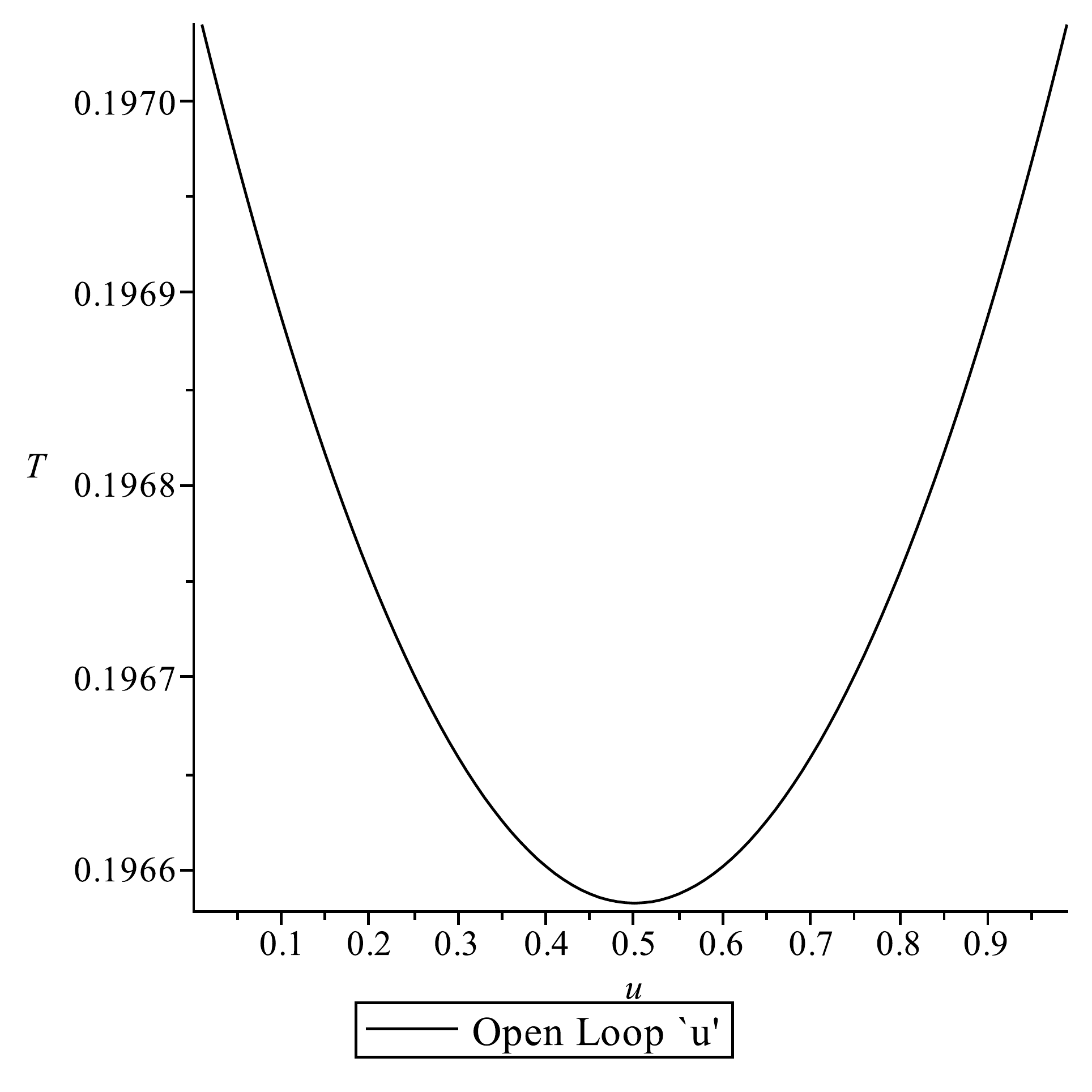}
\caption{Mean delay from leecher to seeder vs $u$}\label{no-choke-fig}
\end{figure}
~\\

To consider both a high and low uplink bandwidth cases, we use
superscripts $\hi$ and $\lo$ on the parameters above, 
and let $x_k := x^{\hi}_k + x^{\lo}_k$ for $k\in\{l,r,N-1,s\}$.  
We assume all peers know which segment is the rare one.
We also assume uplink bandwidth discrepancies
{\em only} eliminate  the \textit{swap} terms $x_r^{\hi} x_{N-1}^{\lo}$
through choking, i.e.,
a high-uplink peer will trade with a low-uplink peer if that 
low-uplink peer is providing the rare segment.
Since all transactions are one-way when leechers are involved, we have the dynamics:
\begin{align}
\dot{x}_l^{\lo}   =  \lambda_l^{\lo} -[\beta_r x_r +\beta_{N-1}x_{N-1}
+ (u \beta_r + (1-u)\beta_{N-1})x_s] x_l^{\lo}\\
\dot{x}_l^{\hi}  = \lambda_l^{\hi} -[\beta_r x_r +\beta_{N-1}x_{N-1}
+ (u \beta_r + (1-u)\beta_{N-1})x_s] x_l^{\hi},
\end{align}

Leecher dyanmics without uplink considerations above are obtained by adding these two equations 
with $\lambda_k := \lambda^{\hi}_k + \lambda^{\lo}_k$ for $k\in\{l,s\}$. 
The complete set of differential equations describing system behavior are: 
\begin{align}
\dot{x}_r^{\lo}  = &  \beta_r (u x_s +x_r)x_l^{\lo} -(\beta_{N-1}x_s 
+\beta_r x_{N-1})x_r^{\lo}\\
\dot{x}_r^{\hi}  = &  \beta_r (u x_s +x_r)x_l^{\hi} -(\beta_{N-1}x_s 
+\beta_r x_{N-1}^{\hi})x_r^{\hi}\\
\dot{x}_{N-1}^{\lo}  = &  \beta_{N-1} ((1-u) x_s +x_{N-1})x_l^{\lo} 
-\beta_r(x_s + x_r^{\lo})x_{N-1}^{\lo}\\
\dot{x}_{N-1}^{\hi}  = &  \beta_{N-1} ((1-u) x_s +x_{N-1})x_l^{\hi} 
-\beta_r(x_s + x_r)x_{N-1}^{\hi}\\
\dot{x}_s^{\lo}  = & \lambda_s^{\lo} +
(\beta_{N-1}x_s +\beta_r x_{N-1}) x_r^{\lo}
+\beta_r(x_s+x_r^{\lo})x_{N-1}^{\lo} -\delta^{\lo} x_s^{\lo}\\
\dot{x}_s^{\hi}  = & \lambda_s^{\hi} +
(\beta_{N-1}x_s +\beta_r x_{N-1}^{\hi}) x_r^{\hi}
+\beta_r(x_s+x_r)x_{N-1}^{\hi} -\delta^{\hi} x_s^{\hi}
\end{align}
Note how the final terms
involving population products (i.e., Hamer terms \cite{Daley-Gani}) 
on the right-hand-side  differ for 
the high and low uplink population dynamics to account for choking.

For these coupled dynamics involving both high-uplink and low-uplink peers,
we considered a special case with the number of high-uplink peers much smaller than those of low-uplink peers, specifically with:
$\lambda_l^{\lo} = 10$, $\lambda^{\hi}_s = 0.1$, $\lambda_l^{\hi} = 1$, 
$\lambda_s^{\hi}= 0.01$, $\delta^{\lo} = 9$, $\delta^{\hi}=0.9$, $\beta_r=\beta_N-1 = 1$.
The sensitivity of the delay from leecher to seeder for the low-uplink peers was illustrated in
Figure \ref{no-choke-fig}.
The same quantity for the two-type system  is plotted in Figure \ref{high-uplink-fig}, where 
we also plot the continuous globally-rarest-first optimal-control sojourn times. 
\begin{figure}[htbp]
\centering
\includegraphics[width=2.5 in]{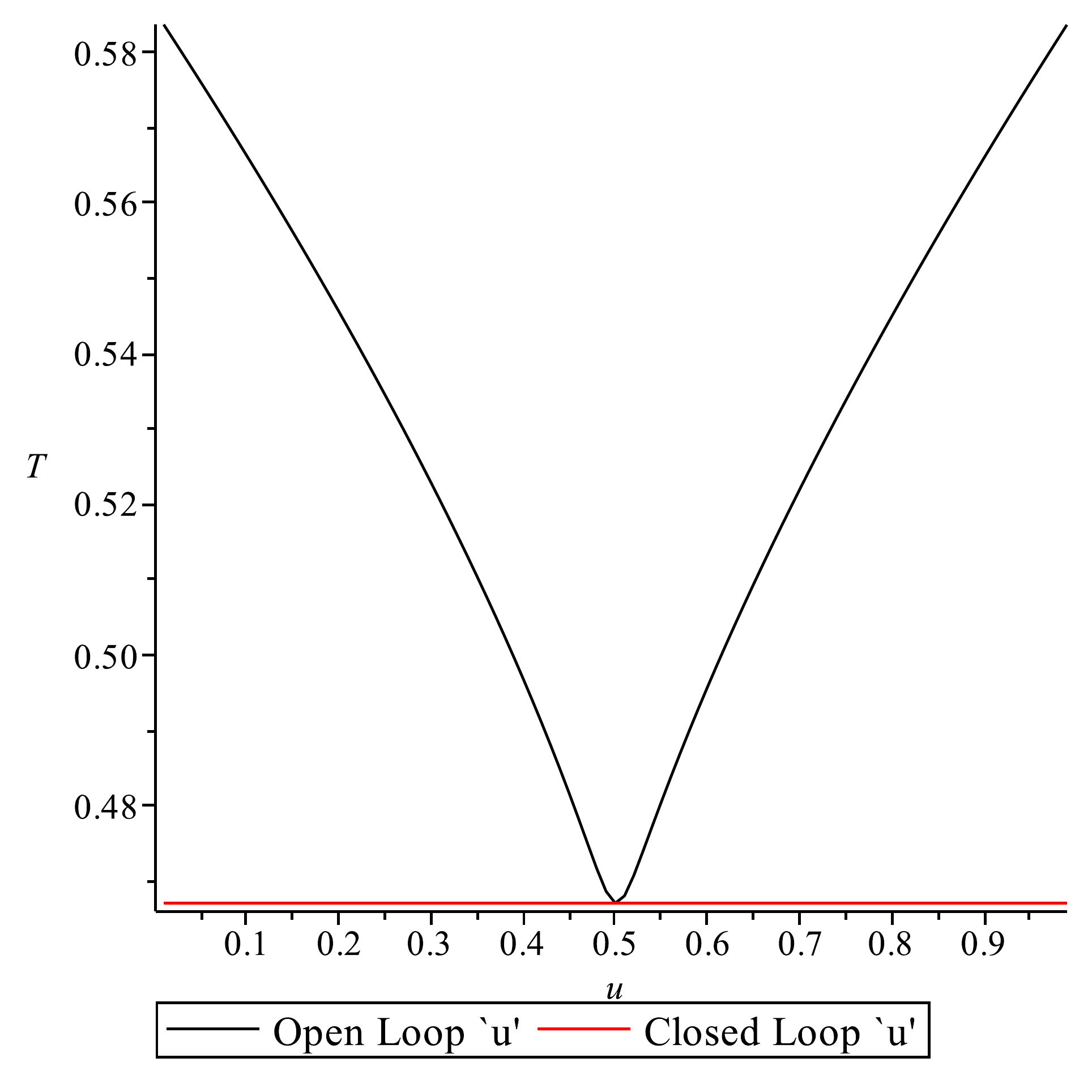}
\caption{Mean delay from leecher to seeder for high-uplink peers vs $u$
with $\delta^{\lo} = 9$ and $\delta^{\hi}=0.9$}\label{high-uplink-fig}
\end{figure}
Note the shape of the curve describing the sojourn time as a function of $u$ is highly affected by the parameter choices. For example, choosing smaller values 
$\delta^{\lo} = 1$ and $\delta^{\hi}=0.1$ produces 
Figure \ref{high-uplink-fig-2}.

\begin{figure}[htbp]
\centering
\includegraphics[width=2.5 in]{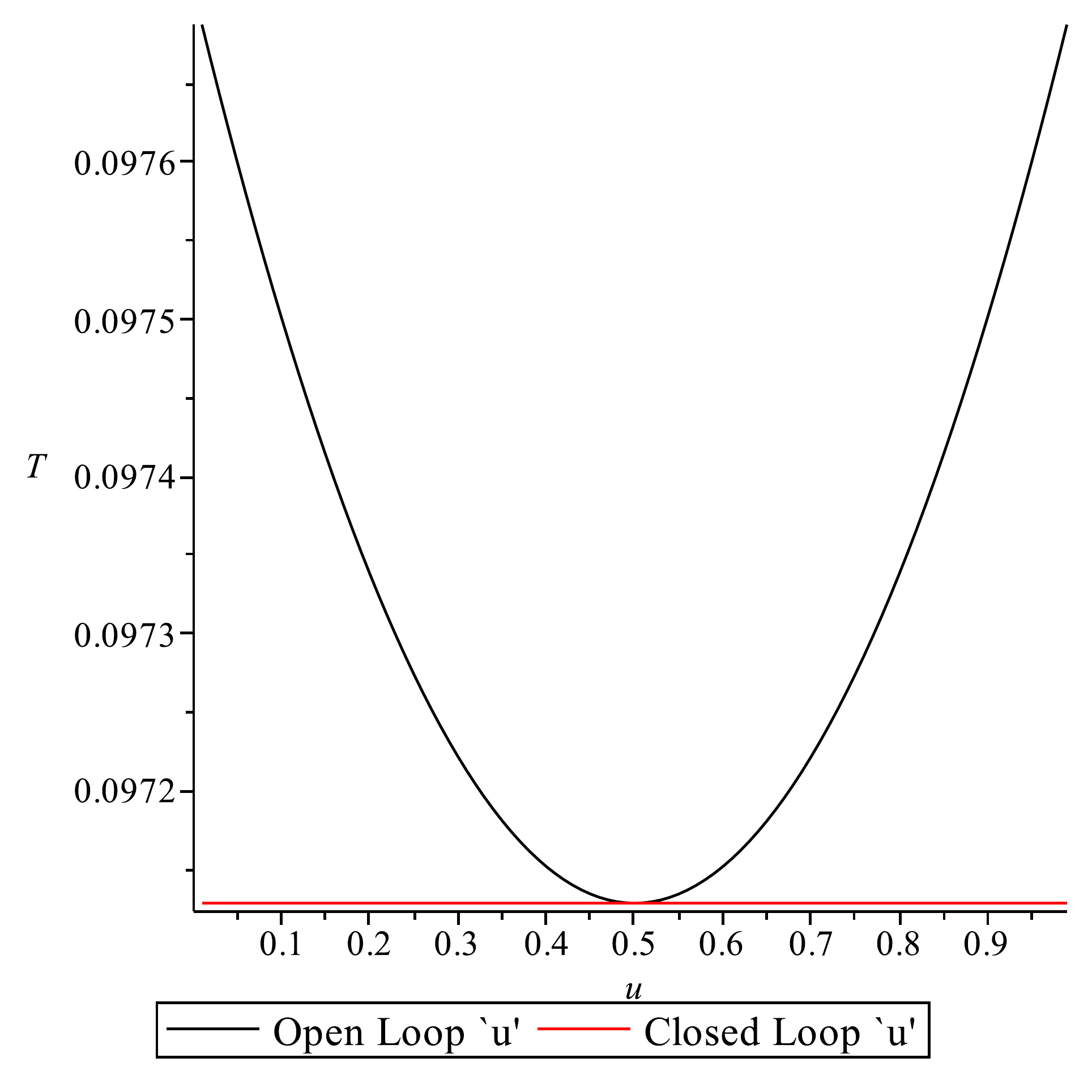}
\caption{Mean delay from leecher to seeder for high-uplink peers vs $u$
with 
$\delta^{\lo} = 1$ and $\delta^{\hi}=0.1$}\label{high-uplink-fig-2}
\end{figure}

The continuous globally rarest first control produces almost identical behavior to the case when $u \equiv 1/2$. This is because both the bang-bang controller and its continuous variation drive the system toward the case when the number of rare-segments holders is identical to the number of ``$N-1$ segments" holders. When this occurs, the values of the three controllers are identical. In computing these plots, the number of rare segment holders was initialized lower than the number of $N-1$ segment holders. This suggests that there is little to be gained from the more complex \textit{rarest first} controllers but a substantial amount 
to be gained in terms of swarm stability by setting $u \neq 1/2$.

\section{Summary}\label{concl-sec}

Ostensibly under BitTorrent incentives, the locally-rarest-first rule's objective is to prevent extinction of certain segments. The file segmentation system itself is intended to extend a peer's time in a swarm and thereby increase their cooperation (swapping activity).
Choking and optimistic unchoking mechanisms  result in a clustering of peers according to their allocated uplinks \cite{Legout06}, i.e.,  peers  will naturally tend to swap with others with similar uplink rates. Unchoking is intended to give peers with low uplinks a chance to increase (rehabilitate) their uplinks. Note how these incentives may be undermined by a large number of seeder peers. When these incentives do not work because many leecher peers are simply unwilling or unable to increase their uplinks and when the present seeders are congested,  unchoking may allow peers to access needed segments even if it means acquiring them at a rate significantly slower than their own allocated uplinks for file-sharing.

Typically, there are a  persistent number of  ``permanent" seeders which exclusively perform server transactions in the swarm that operate ``outside" of these incentives to distribute segments and prevent segment extinction \cite{Bieber06}. These permanent seeders are fostered by: fixed-rate pricing frameworks for Internet access, and limited liability for copyright infringement afforded by  the file segmentation framework itself as well as by third-party swarm discovery
(e.g.,  downloading torrents via  certain web sites). 

Given a positive departure rate for seeders, the presence of permanent seeders is here modeled by the assumption that $\lambda_s>0$, i.e.,  a small but  persistent arrival rate of seeders.

We discussed the optimality properties of the locally-rarest-first segment distribution policy and how to employ different policies that create  relatively rare segments. The leechers are thereby enticed to stay somewhat longer in the swarm and, consequently, cooperate more.

In the local swarm, delays will increase the availability of segments 
that can be exchanged between leechers. Globally, they would cause
leechers to become temporary seeders for more time in other swarms or 
even under certain assumptions in the local one~\cite{Bieber06}.

We note that delays will result in larger average total leecher populations $x_l$ that will consequently create a  lesser burden on the permanent seeder population $x_s$. To reflect the limited uplink capacity of the permanent seeders, we might want to reduce the parameter $\beta$ as $x_l$ increases  so that the $\beta x_l$ factor is fixed, i.e., so that the segment transfer rate $\beta x_l x_s$ reflects these limits.

Finally, we also studied the system
interpreting the two-segment model as
a rare segment and a aggregation of the rest, naturally leading
to different associated uplink parameters. 
We numerically showed how leecher and seeder sojourn times
were affected by the (lumped) model parameters, particularly the proportion of
the population that are seeders, and argued that
taking the control $u \neq 1/2$ leads to good stability properties.

\section{Acknowledgements}
G. Kesidis' work was funded in part by NSF CNS NeTS Grant No. 0915928.

\bibliography{LocalBib}
\bibliographystyle{IEEEtran}

\end{document}